\documentclass[letterpaper, 10 pt, conference]{ieeeconf}
\IEEEoverridecommandlockouts
\usepackage[left=2cm,right=2cm,top=2cm,bottom=2cm]{geometry}
\usepackage{amsmath,amssymb,amsfonts}
\usepackage{algorithmic}
\usepackage{graphicx}
\usepackage{textcomp}
\usepackage{xcolor}

\def\usetikz{0} 
\if\usetikz1
	\usepackage{tikz}
	\usepackage{pgfplots}
	\pgfplotsset{compat=1.14}
	\usepgfplotslibrary{groupplots}
	\usetikzlibrary{positioning}
	\usetikzlibrary{external} 
	\tikzset{
	    png export/.style={
			external/system call/.add={}
	        && convert -density 600 -transparent white "\image.pdf" "\image.png", 
	        /pgf/images/external info,
	        /pgf/images/include external/.code={%
	            \includegraphics
	                [width=\pgfexternalwidth,height=\pgfexternalheight]
	                {##1.png}%
	        },
		} 
	}
\fi

\usepackage{upgreek}
\usepackage{float} 
\usepackage{enumerate}

\usepackage{amsthm}
\usepackage{cleveref}

\crefname{lem}{Lemma}{Lemmas}
\newtheorem{thm}{Theorem} 
\crefname{thm}{Theorem}{Theorems}
 
\crefname{cor}{Corollary}{Corollaries}

\crefname{prop}{Proposition}{Propositions}
\theoremstyle{definition}
\newtheorem{rmk}{Remark}
\crefname{rmk}{Remark}{Remarks}

\crefname{defn}{Definition}{Definitions}
\newtheorem{ex}{Example}
\crefname{ex}{Example}{Examples}
\newtheorem{assum}{Assumption}
\crefname{assum}{Assumption}{Assumptions}
\crefname{appendix}{Appendix}{Appendices}
\crefname{section}{Section}{Sections}
\crefname{table}{Table}{Tables}

\newcommand{\setreal}{\mathbb{R}}

\newcommand{\transp}{\mathsf{T}}

\newcommand{\taumin}{\underline{\tau}}
\newcommand{\taumax}{\overline{\tau}}
\newcommand{\Pdiag}{\bar{P}}
\newcommand{\Rdiag}{\bar{R}}

\DeclareMathOperator*{\diag}{diag}

\newcommand{\estnot}[1]{{\hat{#1}}}
\newcommand{\yest}{\estnot{v}}
\newcommand{\west}{\estnot{w}}
\newcommand{\thetaest}{\estnot{\theta}}

\newcommand{\vest}{\estnot{v}}

\newcommand{\trunot}[1]{#1}
\newcommand{\ytru}{\trunot{v}}
\newcommand{\wtru}{\trunot{w}}
\newcommand{\thetatru}{\trunot{\theta}}
\newcommand{\utru}{\trunot{u}}

\newcommand{\vtru}{\trunot{v}}

\newcommand{\nernsttru}{\trunot{\nu}}

\newcommand{\col}{\mathrm{col}}

\newcommand{\gain}{\gamma}

\newcommand{\bv}{a}

\newcommand{\fw}{g}
\newcommand{\fvvirt}{\tilde{f}}

\newcommand{\vhalf}{\rho}
\newcommand{\slope}{\kappa}
\newcommand{\mean}{\zeta}
\newcommand{\std}{\chi}

\newcommand{\nernst}{\nernsttru}

\newcommand{\lowerbound}{\underline{\sigma}}
\newcommand{\upperbound}{\overline{\sigma}}

\newcommand{\Ytru}{V}

\newcommand{\nv}{{n_v}}

\newcommand{\nwj}{{n_w^j}}

\newcommand{\ninp}{{n_u}}
\newcommand{\ny}{{n_v}}
\newcommand{\ntheta}{{n_\theta}}
\newcommand{\nthetaj}{{n_\theta^j}}

\newcommand{\Na}{{\textrm{Na}}}
\newcommand{\K}{{\textrm{K}}}

\newcommand{\GABA}{{\textrm{G}}}
\newcommand{\Leak}{{\textrm{leak}}}

\newcommand{\ion}{{\rm{ion}}}

\newcommand{\syn}{{\rm{syn}}}

\title{\LARGE \bf
Distributed online estimation of biophysical neural networks*
}

\author{Thiago B. Burghi$^{1}$, Timothy O'Leary$^{1}$ and Rodolphe Sepulchre$^{1}$
	\thanks{*The research leading to these results has received funding from the European Research Council under the ERC grant agreement FLEXNEURO n.716643.}
	\thanks{$^{1}$Thiago B. Burghi, Timothy O'Leary  and Rodolphe Sepulchre are with the Department of Engineering, University of Cambridge, Cambridge CB2 1PZ, United Kingdom. E-mails: {\tt\small tbb29@cam.ac.uk}, {\tt\small tso24@cam.ac.uk}, {\tt\small rs771@cam.ac.uk}.}
}

\begin{document}

\maketitle
\thispagestyle{empty}
\pagestyle{empty}

\begin{abstract}

In this work, we propose a distributed adaptive observer for a class of nonlinear networked systems inspired by biophysical neural network models. Neural systems learn by adjusting intrinsic and synaptic weights in a distributed fashion, with neuronal membrane voltages carrying information from neighbouring neurons in the network. We show that this learning principle can be used to design an adaptive observer based on a decentralized learning rule that greatly reduces the number of observer states required for exponential convergence of parameter estimates. This novel design is relevant for biological, biomedical and neuromorphic applications.
\end{abstract}

\section{Introduction}

With the improvement of neural recording technology, it may soon be possible to concurrently monitor the membrane potential of hundreds of interconnected neurons in a living brain \cite{knopfel_optical_2019}.  This high-resolution data opens up new possibilities for the development of real-time closed-loop interventions aimed at treating disorders of neural excitability such as epilepsy and Parkinson's  \cite{tang_colloquium_2018}. The capability to effectively monitor and control spiking systems also impacts the nascent field of neuromorphic engineering \cite{ribar_neuromodulation_2019}.

Good closed-loop control design often requires reliable model estimates, and hence any method aimed at controlling neural activity is bound to involve the estimation of neuronal models, which is a nontrivial task. Many techniques have been proposed for \textit{batch-mode} or \textit{offline} estimation of neuronal dynamics, see for instance \cite{huys_efficient_2006,druckmann_novel_2007,meliza_estimating_2014,burghi_feedback_2021}.
However, living brain systems are adaptive \cite{sorrell_brainmachine_2021}, and thus online estimation approaches are necessary, especially if real-time applications are involved.

To meet this demand, an adaptive observer-based approach for online estimation of conductance-based neural networks was recently proposed in \cite{burghi_adaptive_2021}. The adaptive observers, inspired by \cite{zhang_adaptive_2001} and \cite{farza_adaptive_2009}, are rooted in the familiar Recursive Least Squares (RLS) algorithm \cite{astrom_adaptive_2008}, and allow for approximately tracking slowly time-varying parameters. One limitation of RLS-based adaptive observers is the rapid  increase in observer states with respect to the number of parameters. More observer states require more computing power, which might become critical when attempting to perform online estimation of large neuronal network models containing thousands of parameters. In this paper, we propose a distributed version of the linear-in-the-parameters adaptive observer from \cite{burghi_adaptive_2021} that results in a scalable algorithm for online parameter estimation of biophysical neural models. The proposed modification, which echoes the diagonal RLS-like update rule of  \cite{vahidi_recursive_2005}, greatly reduces the number of adaptive observer states. We show that for neuronal network models, the proposed adaptive observer becomes distributed over individual neuronal membrane currents and also over neurons in the network. We analyse the adaptive observer using contraction theory \cite{lohmiller_contraction_1998}, and show that a strengthened  persistent excitation condition is sufficient for consistent convergence of the parameter estimates.

The paper is organized as follows: in \cref{sec:background}, we pose the problem from an abstract point of view and recall the observer from \cite{burghi_adaptive_2021}. In \cref{sec:distributed_observer}, we introduce and analyse the modified distributed observer. In \cref{sec:CB_models}, we use the observer to estimate conductance-based biophysical neural networks. In \cref{sec:discussion}, we discuss the relevance of this work in neuroscience, as well as future research directions.

\textit{Notation:} 
We write $I_n$ for the $n \times n$ identity matrix, and $I$ when $n$ is obvious from the context. 
For two column vectors $x$ and $y$, we write $\col(x,y) := (x^\transp,y^\transp)^\transp$. For a matrix $A \in \setreal^{n\times n}$, $\lambda_{\max}(A)$ denotes the largest eigenvalue of $A$.
For a vector-valued function $f:\setreal^{n_1} \times \setreal^{n_2} \to \setreal^m$, we write $\partial_x f(x,y) \in \setreal^{m \times n_1}$ for the Jacobian of $f(x,y)$ with respect to $x$. We write $A \succeq B$ ($A \succ B$) if  $A-B$ is a positive-semidefinite (positive-definite) matrix. 

\section{Background}
\label{sec:background} 

We consider nonlinear state-space systems of the form
\begin{subequations}
	\label{eq:system} 
	\begin{align}
		\dot{v} &= \sum_{j=1}^m 
		\Phi_j^\transp(v,w^j,u)
		\theta^j + \bv(v,w,u) \\
		\label{eq:dw_true} 
		\dot{w}^j &= \fw_j(v,w^j)
	\end{align}
\end{subequations}
for $j=1,\dotsc,m$. Here, $v \in \setreal^\nv$ is a state vector, which is also the output of the system; $w = \col(w^1,\dotsc,w^m)$ is an internal dynamics state vector, with $w^j \in \setreal^{\nwj}$; $u\in\setreal^\ninp$ a control input vector; and $\theta = \col(\theta^1,\dotsc,\theta^m)$ is a parameter vector, with $\theta^j \in \setreal^{{\nthetaj}}$. The matrices $\Phi_j(v,w^j,u) \in \setreal^{\nthetaj\times\nv}$ and the vectors $\bv(v,w,u)\in\setreal^{\nv}$ and $\fw_j(v,w^j)\in\setreal^{\nwj}$ are assumed to be continuously differentiable in their arguments. We will also use the more compact notation
\begin{equation*}
	\Phi^\transp(v,w,u) := 
	\begin{bmatrix}
		\Phi_1^\transp(v,w^1,u) & \dotsc & \Phi^\transp_m(v,w^m,u)
	\end{bmatrix}
\end{equation*}
and
\begin{equation*}
	\fw(v,w) := \col(
		\fw_1(v,w^1),\dotsc,\fw_m(v,w^m)) \;.
\end{equation*}
The role of an adaptive observer is to provide an online estimation of the states and the parameters of the system from measurements of the input $u(t)$ and output $v(t)$.

The structure of  \eqref{eq:system} is motivated by models of neuronal dynamics \cite{izhikevich_dynamical_2007}. The state $v$ represents a vector of membrane voltages in a neural network, while $w$ represents a vector of \textit{gating variables} that dictate ion channel and synaptic dynamics. This specific application is discussed in \cref{sec:CB_models}.

In our problem formulation, we assume that the trajectories of the system \eqref{eq:system} evolve in a compact positively invariant set, and that the internal dynamics of \eqref{eq:system} are exponentially contracting \cite{lohmiller_contraction_1998}, uniformly in $v$:

\begin{assum}
	\label{assum:true_invariant_set} 
		There exists a compact set $\Ytru \times U$
		such that $\{\ytru(t),\utru(t)\} \in \Ytru \times U$
		for all $t \ge 0$.
\end{assum}

\begin{assum}
	\label{assum:int_dyn_contraction}
	For each $j = 1,\dotsc,m$, there exists a compact convex set $W_j$ which is positively invariant with respect to \eqref{eq:dw_true}, uniformly in $\ytru$ on $\setreal^\nv$.	Furthermore, there exist a symmetric matrix $M_j(t)=\Theta_j(t)^\transp\Theta_j(t)$ such that
		$\lowerbound I \preceq M_j(t) \preceq \upperbound I$
	for some $\lowerbound,\upperbound>0$, and a contraction rate $\lambda_j > 0$ such that the generalized Jacobian
	\begin{equation}
		\label{eq:Mint}
		F_j := (\dot{\Theta}_j + \Theta_j \partial_{w^j} \fw_j(\ytru,\wtru^j)) \Theta_j^{-1}
	\end{equation}
	satisfies
	\begin{equation}
	\label{eq:contraction} 
		F_j+F_j^\transp \preceq -\lambda_j I
	\end{equation}
	for all $\{\ytru,w^j\} \in \setreal^{\ny} \times W_j$ and
	all $t \ge 0$.
\end{assum}

In \cite{burghi_adaptive_2021}, an adaptive observer-based approach (inspired by the earlier designs of \cite{zhang_adaptive_2001} and \cite{farza_adaptive_2009}) was proposed to estimate the parameter vector $\theta = \col(\theta^1,\dotsc,\theta^m)$ in real-time. In the present work, our point of departure is the adaptive observer given by
\begin{subequations}
	\label{eq:adaptive_observer}
	\begin{align}
		\nonumber
	  	\dot{\yest} &= \Phi^\transp(\ytru,\west,\utru)\thetaest 
		+ \bv(\ytru,\west,\utru)
	  	+ \gain(I +\Psi^\transp P \Psi)(\ytru-\yest)
	  	\\
		\label{eq:adaptive_states}
	  	\dot{\west} &= \fw(\ytru,\west) 
	  	\\
		\nonumber
	  	\dot{\thetaest} &= \gain P \, 
	  	\Psi \, (\ytru-\yest)
	\end{align}
	where the matrices $P$ and $\Psi$ 
	evolve according to
\begin{align}	
	  	\label{eq:dPsi}
		  	\dot{\Psi} &= 
			- \gain \Psi + \Phi(\ytru,\west,\utru) \;,
	  	\\
	  	\label{eq:dP} 
	  	\dot{P} &= \alpha P -  
	  	\alpha P \, \Psi \Psi^\transp P,
	  	\quad\quad P(0) \succ 0,
\end{align}
\end{subequations}
with $\gamma>\alpha > 0$. Under \cref{assum:true_invariant_set,assum:int_dyn_contraction}, and under the persistent excitation (PE) assumption that
\begin{equation}
\label{eq:PE_condition} 
\exists \, T>0, \; \forall t \ge 0 \;  : \;
 \underline{\delta} I \preceq \int_{t}^{t+T}
	\Psi(\tau) \Psi^\transp(\tau) d\tau \preceq \overline{\delta} I
\end{equation}
for some $\overline{\delta},\underline{\delta}>0$,
it can be shown that the adaptive observer state vector $\col(\hat{v}(t),\hat{w}(t),\hat{\theta}(t))$ converges to $\col(v(t),w(t),\theta)$ exponentially fast as $t \to \infty$ (see \cite[Theorem 1]{burghi_adaptive_2021} and its proof).

The adaptive observer \eqref{eq:adaptive_observer} relies on the $\ntheta \times \ntheta$ matrices $P(t)$ and $\Psi\Psi^\transp$ to update the parameter estimates $\thetaest$. 
From a computational point of view, when the number $\ntheta = \sum_{j=1}^m \nthetaj$ is large, updating the $(\ntheta)^2$ states of $P(t)$ becomes costly. In this paper, we are interested in redesigning the adaptive observer above so as to decrease the number of required observer states in $P(t)$. 

We will explore the simple idea that the matrix $\Psi\Psi^\transp$ can be approximated by its (block) diagonal elements $\Psi_j\Psi_j^\transp$ under suitable assumptions.
This will lead to a \textit{decoupled} version of the update rule for each  component $\theta^j$ of $\theta$ which is in addition \textit{distributed} with respect to the internal dynamics states $w^j$. 

\section{Distributed adaptive observer}
\label{sec:distributed_observer} 

We now consider the adaptive observer design given by
\begin{subequations}
	\label{eq:adaptive_observer_dist}
	\begin{align}
	\label{eq:dv_dist}
	  	\dot{\yest} &= 
		\sum_{j=1}^m \Phi^\transp_j(v,\west^{j},u)\thetaest^{j} 
		+ \bv(\ytru,\west,\utru)
	  	\\ \nonumber
	  	 & \hspace{4em} + \big(\gamma_0 I + \gamma_j \sum_{j=1}^m \Psi_j^\transp P_j \Psi_j\big)(\ytru-\yest)
	  	\\
	  	\label{eq:dw}
	  	\dot{\west}^j &= \fw_j(\ytru,\west^j)
	  	\\
	  	\label{eq:dtheta_dist}
	  	\dot{\thetaest}^j &= \gain_j P_j \, 
	  	\Psi_j \, (\ytru-\yest)
	\end{align}
	where 
$\gamma_0,\gamma_1,\dotsc,\gamma_m>0$ are constant gains, and the matrices $P_j$ and $\Psi_j$ evolve according to
	\begin{align}
	  		\label{eq:dPsi_dist}
	  		\dot{\Psi}_j &= - \gamma_j \Psi_j 
	  		+ \Phi_j(\ytru,\west^j,u), &\quad \Psi_j(0) = 0,
	  		\\
	  		\label{eq:dP_dist} 
	  		\dot{P}_j &= \alpha_j P_j -  
	  		\alpha_j P_j \, \Psi_j \Psi^\transp_j P_j, &\quad P_j(0)\succ 0,
	\end{align}	
\end{subequations}
for all $j=1,\dotsc,m$ (there is no loss of generality in the choice of  $\Psi_j(0)$ above). 

The adaptive observer  \eqref{eq:adaptive_observer_dist} is obtained 
from \eqref{eq:adaptive_observer} by 
ignoring the off-diagonal terms in $\Psi\Psi^\transp$, enforcing a block-diagonal structure in the matrix $P(t)$, and allowing for different learning gains $\gamma_j>0$ and forgetting rates $\alpha_j>0$.
In the new design, the adaptation law given by \eqref{eq:dtheta_dist}-\eqref{eq:dP_dist} is distributed with respect to each ``regressor'' $\Phi_j^\transp$ and its associated internal dynamics \eqref{eq:dw}. The total number of states in the matrices $P_j$ for $j=1,\dotsc,m$ is now reduced to $m\times \sum_{j=1}^m (\nthetaj)^2 \le (\ntheta)^2$. In addition, the learning gains and forgetting rates can now be chosen independently, which might be beneficial to account for different timescales in the internal dynamics components \eqref{eq:dw_true}. These timescales correspond to the contraction rates $\lambda_j$ in \cref{assum:int_dyn_contraction}.

\begin{rmk}
The distributed adaptive observer \eqref{eq:adaptive_observer_dist} was designed by modifying the design of \cite{burghi_adaptive_2021}, which exactly solves a Recursive Least Squares (RLS) problem with exponential forgetting \cite[Chapter 2]{astrom_adaptive_2008}. In a similar spirit, a modification of the RLS algorithm based on a diagonal update rule was previously proposed by \cite{vahidi_recursive_2005}.
However, not much is known about the convergence properties of such diagonal update rules, and they do not seem to have been considered in the context of adaptive observers.
\end{rmk}

\begin{rmk}
Many authors have studied the related problem of \textit{directional forgetting} in  RLS estimation, e.g. \cite{bittanti_convergence_1990,parkum_recursive_1992,fraccaroli_new_2015}.  The adaptive observer above is not based on directional forgetting methods, however, since such methods still consider off-diagonal terms of $\Psi\Psi^\transp$ in the update rule. It is also worth mentioning that the diagonal update rule above is reminiscent of the idea of \textit{splitting across features} when solving regression problems with the Alternating Direction Method of Multipliers \cite{boyd_distributed_2010}.
\end{rmk}

\begin{rmk}
While the rather general PE condition \eqref{eq:PE_condition} is a sufficient condition for the consistent exponential convergence of the adaptive observer \eqref{eq:adaptive_observer}, the same condition may not be sufficient for exponential convergence of the adaptive observer with the diagonal update rule \eqref{eq:adaptive_observer_dist}. The reason can be illustrated by considering the simplified linear model
	\[
		y = \gamma \Psi(t)^\transp \theta
	\]
with known $\Psi(t)$. It is well-known that the update rule 
	\begin{equation}
		\label{eq:RLS_simple} 
		\dot{\thetaest} = P \Psi (y-\gamma\Psi^\transp \thetaest)
	\end{equation}
	with $P$ given by \eqref{eq:dP} ensures that $\thetaest(t) \to \theta$ exponentially fast as long as the PE condition \eqref{eq:PE_condition} holds. Indeed, the PE condition ensures that $P(t)$ is uniformly positive definite and bounded above \cite{zhang_adaptive_2001}, leading to the Lyapunov function candidate
		$V(t,\tilde{\theta}) = \tilde{\theta}^\transp P^{-1}(t) \, \tilde{\theta}$	
	where $\tilde{\theta} := \theta - \thetaest$. This function satisfies
	\[
		\dot{V}(t,\tilde{\theta}) = - \tilde{\theta}^\transp ( \alpha P^{-1} + (2\gamma-\alpha)\Psi \Psi^\transp )\tilde{\theta} \;,
	\]
	which is uniformly negative definite for $\gamma>\alpha/2$. 
	Suppose now that we replace the update rule \eqref{eq:RLS_simple} by its diagonal counterpart 
	\[
		\dot{\thetaest} =  \bar{P} \Psi(y-\gamma\Psi^\transp \hat{\theta})\;,
	\]
	with 
	$
		\bar{P} := \diag\{P_1,\dotsc,P_m\}
	$
	and 	$P_j(t)$ given by \eqref{eq:dP_dist}, where for simplicity we take $\alpha_j = \alpha$. Persistent excitation with respect to each $\Psi_j\Psi_j^\transp$ still ensures that $\bar{P}(t)$ is uniformly positive definite and bounded above, but PE alone is not sufficient to ensure the uniform negative definiteness of the derivative of the Lyapunov function candidate $\bar{V}(t,\tilde{\theta}) = \tilde{\theta}^\transp \bar{P}^{-1}(t) \,\tilde{\theta}$.
\end{rmk}

The previous remark illustrates that the convergence analysis of the distributed observer requires further investigation. In the same vein as \cite{aeyels_exponential_1999,jouffroy_relaxed_2003}, we shall give up on negative definiteness (or semidefiniteness) of the derivative of the Lyapunov function. To state our final assumption, we define the matrices
\begin{equation}
	\label{eq:gain_matrices} 
	\begin{split}
	\Pdiag &:= \diag\{P_1,\dotsc,P_m\} \;, \\
	\Gamma &:= \diag\{\gamma_1 I_{n_\theta^1} ,\dotsc,\gamma_m I_{n_\theta^m} \} \;, \\
	A &:= \diag\{\alpha_1 I_{n_\theta^1} ,\dotsc, \alpha_m I_{n_\theta^m}\}\;, \\
	D &:= \diag\{\Psi_1\Psi_1^\transp,\dotsc, \Psi_m\Psi_m^\transp\}\;,
	\end{split}
\end{equation}
where the $P_j$ above come from \eqref{eq:dP_dist}.

\begin{assum}
	\label{assum:required_assumption} 
	For all $v(t)$, $u(t)$, there exists a
	$T>0$ such that for all $t \ge 0$, the following hold:
	\begin{enumerate}[(i)]
		\item For each $j=1,\dotsc,m$, we have
		\[
		\int_t^{t+T} \Psi_j(\tau)\Psi_j(\tau)^\transp d\tau \succeq \delta_j
		\]
		for some $\delta_j>0$.
		\item Let $\underline{\alpha}=\min\{\alpha_1,\dotsc,\alpha_m\}$. Then there exists a $\beta > 0$ such that 		
		\begin{equation*}
	\begin{split}
	\frac{1}{T}
	\int_t^{t+T} 
	\lambda_{\max}(
	AD & + \gamma_0 \Psi\Psi^\transp - \Gamma \Psi \Psi^\transp - \Psi \Psi^\transp \Gamma)	d\tau
	\\
	&\le 
	(\underline{\alpha} - \beta)\times{\min_{j} \delta_j \alpha_j e^{-2\,\alpha_j,T} } \;.
	\end{split}
	\end{equation*}		
	\end{enumerate}
\end{assum}


We can now state our main theoretical result:
%

\begin{thm}
	\label{thm:convergence}
	Consider the true system \eqref{eq:system} and the adaptive observer \eqref{eq:adaptive_observer_dist}. Under \cref{assum:true_invariant_set,assum:int_dyn_contraction,assum:required_assumption},
	for any $\vest(0) \in \setreal^{\ny}$, $\west(0) \in \prod_j W_j$, and 	$\thetaest(0) \in \setreal^{\ntheta}$, 
	we have	
	\[
		\col(\vest(t),\west(t),\thetaest(t)) \to \col(\vtru(t),\wtru(t),\thetatru)
	\]
	exponentially fast as $t \to \infty$.
\end{thm}

\begin{proof}
See \cref{sec:proof_convergence}.
\end{proof}

\begin{rmk}
	\cref{assum:required_assumption} part (\textit{ii}) shows that to promote consistent parameter estimation, we can pick $\Gamma = \gamma_0 I$ with $\gamma_0 \gg \alpha_j$ for all $j$: then $AD+ \gamma_0 \Psi\Psi^\transp - \Gamma \Psi \Psi^\transp - \Psi \Psi^\transp \Gamma \approx -\gamma_0 \Psi\Psi^\transp \preceq 0$. However, the freedom in choosing different $\gamma_j$ allows us to be more strategic in terms of how to choose those gains to ensure convergence of the observer, especially when some prior information about the system is available. This is illustrated in \cref{sec:numerical}.
\end{rmk}

\section{Application to biophysical neural networks}
\label{sec:CB_models}

In the context of biophysical neural network models, each component $v_i$ of the vector $v \in \setreal^\nv$ represents the membrane potential of a single neuronal cell. In each of these cells, the membrane potential evolves according to
\begin{equation}
	\label{eq:cb_network}
	c_i\dot{v}_i = -\sum_{\ion\in\mathcal{I}} I^\ion_{i} - \sum_{\syn\in\mathcal{S}} \sum_{k \neq i} I^\syn_{i,k} -\mu^\Leak_i(v_i-\nernst^\Leak) + u_i,
\end{equation}
where $c_i>0$ is a capacitance,
\begin{equation}
	\label{eq:ion_currents} 
	I^\ion_{i} = \mu^\ion_{i} (m^\ion_{i})^{p^\ion} (h^\ion_{i})^{q^\ion} (v_i - \nernst^\ion)
\end{equation}
are \textit{intrinsic ionic currents}, 
\begin{equation}
	\label{eq:syn_currents} 
	I^\syn_{i,k} = \mu^\syn_{i,k} s^\syn_{i,k} (v_i - \nernst^\syn)
\end{equation}
are \textit{synaptic currents}, and 
$\mu^\Leak_i(v_i-\nernst^\Leak)$ is a leak current. The set $\mathcal{I}$ collects ionic current types, while $\mathcal{S}$ collects synaptic current types. The scalars $\mu^\ion_{i}>0$ and $\mu^\syn_{i,k} > 0$ are \textit{intrinsic} and  \textit{synaptic maximal conductances}, and 
the scalars $\nernst^\ion \in \setreal$ and $\nernst^\syn \in \setreal$ are \textit{intrinsic} and \textit{synaptic Nernst potentials}, respectively. 
Finally, the scalars $m^\ion_{i},h^\ion_{i}\in(0,1)$ and $s^\syn_{i,k}\in(0,1)$ are \textit{intrinsic} and \textit{synaptic gating variables}, respectively. Those gating variables modulate the intensity of the currents traversing the neuronal membrane, according to the voltage-dependent dynamics 
\begin{subequations}
	\label{eq:cb_internal} 
	\begin{align}	
		\label{eq:m_dyn}
		\tau_m^\ion(v_i)\,\dot{m}^\ion_{i} &= -m^\ion_{i} + \sigma^\ion_m(v_i) \\
		\label{eq:h_dyn}
		\tau_h^\ion(v_i)\,\dot{h}^\ion_{i} &= -h^\ion_{i} + \sigma^\ion_h(v_i) \\
		\label{eq:s_dyn} 
		\dot{s}^\syn_{i,k} &= a_\syn\sigma^\syn(v_k)(1-s^\syn_{i,k}) - b_\syn s^\syn_{i,k}		
	\end{align}
\end{subequations}
where $\tau(\cdot)$ are bell-shaped\footnote{The results easily extend to other forms of time-constant functions and sigmoidal functions.} functions of the form $\tau(v)=\taumin + (\taumax-\taumin)\exp(-(v-\mean)^2/\std^2)$, with $\taumin,\taumax>0$, where $\sigma(\cdot)$ are sigmoids of the form $\sigma(v) = (1+\exp\left(-(v-\vhalf)/\slope\right))^{-1}$, and where $a_\syn,b_\syn > 0$.

For $i,k = 1,\dotsc,\nv$, the biophysical neural network model given by \eqref{eq:cb_network}-\eqref{eq:cb_internal} can be put in the parametric form \eqref{eq:system}. We illustrate this fact by means of an example in which the components of the parameter vector $\theta$ are given by the maximal conductances $\mu_i^\ion$ and $\mu_{i,k}^\syn$, which constitute key parameters dictating the behavior of the network (see \cite{burghi_adaptive_2021} for more general parametrizations).
\begin{ex}
	\label{ex:HH_network} 
	The Hodgkin-Huxley (HH) biophysical model first introduced by \cite{hodgkin_quantitative_1952} contains a sodium and a potassium intrinsic current, so that $\mathcal{I} = \{\Na,\K\}$. Here we consider two HH neurons interconnected bidirectionally by means of a GABA-type inhibitory synapse (abbreviated by $\GABA$), so that $\mathcal{S} = \{\GABA\}$. We parameterize the model according to
	\[
		\theta = \col(\theta^\Na,\theta^\K,\theta^\GABA)\;,
	\]
	where 
	\[\theta^\Na = (\mu^\Na_1,\mu^\Na_2)^\transp,\;\theta^\K = (\mu^\K_1,\mu^\K_2)^\transp,\;\theta^\GABA = (\mu^\GABA_{1,2},\mu^\GABA_{2,1})^\transp \;.\]
	The voltage dynamics \eqref{eq:cb_network} of this model can be expressed as
	\begin{equation*}
	\dot{v} = \Phi_\Na(v,w^\Na)\theta^\Na
	+ \Phi_\K(v,w^\K)\theta^\K
	+ \Phi_\GABA(v,w^\GABA)\theta^\GABA
	+ \bv(v,u)
	\end{equation*}
	with the internal states
	\begin{align*}
	w^\Na &= (m^\Na_1,h^\Na_1,m^\Na_2,h^\Na_2)^\transp,
	w^\K = (m^\K_1,m^\K_2)^\transp, \text{ and } \\
	w^\GABA &= (s^\GABA_{1,2},s^\GABA_{2,1})^\transp,
	\end{align*}		
	the system matrices
	\begin{equation*}
	\resizebox{\hsize}{!}{$
	\begin{aligned}
	\Phi_\Na(v,w^\Na) &= 
	-\diag\left ( (w^\Na_1)^3w^\Na_2(v_1-\nernst^\Na)\,,\, -(w^\Na_3)^3w^\Na_4(v_2-\nernst^\Na) \right),
\\
	\Phi_\K(v,w^\K) &= 
	-\diag\left(
	(w^\K_1)^4(v_1-\nernst^\K)\,,\,
	(w^\K_2)^4(v_2-\nernst^\K)
	\right),
	\\
	\Phi_\GABA(v,w^\GABA) &= 
	-\diag\left(
		w^\GABA_1(v_1-\nernst^\GABA)\,,\,
		w^\GABA_2(v_2-\nernst^\GABA)
	\right),
	\end{aligned}
		$}
	\end{equation*}
	and the known vector
	\begin{equation*}
		\bv(v,u) = 
	\begin{pmatrix}
		-\mu_1^\Leak(v_1-\nernst^\Leak)
		+ u_1 \\
		-\mu_2^\Leak(v_2-\nernst^\Leak)
		+ u_2
	\end{pmatrix}.
	\end{equation*}
	With the choice of internal states above, the internal dynamics \eqref{eq:cb_internal} can be expressed as \eqref{eq:dw_true}, and hence the model is in the form \eqref{eq:system} (for clarity, we have replaced indexing by $j = 1,\dotsc,m$ with indexing by $\ion \in \mathcal{I}$ and $\syn \in \mathcal{S}$).
\end{ex}

It can be shown that any conductance-based model \eqref{eq:cb_network}-\eqref{eq:cb_internal} satisfies \cref{assum:true_invariant_set,assum:int_dyn_contraction}; in particular the internal dynamics \eqref{eq:cb_internal} are exponentially contracting (see \cite{burghi_adaptive_2021} for rigorous proofs). Hence,  \eqref{eq:adaptive_observer_dist} can be used to estimate the maximal conductances of any conductance-based model.

\cref{ex:HH_network} shows an interesting feature of biophysical neural network models: the matrices $\Phi_j(v,\cdot)$ are diagonal, with the $i^\mathrm{th}$ diagonal elements depending only on the voltage and internal states of the $i^\mathrm{th}$ neuron. This feature allows for distributing parameter estimation over the neurons in the network, in addition to distributing it over individual membrane currents. This is illustrated next:

\begin{ex}
	\label{ex:network_distribution} 
	Consider the two-neuron HH network of \cref{ex:HH_network}. With the distributed observer \eqref{eq:adaptive_observer_dist}, the sodium maximum conductance update rule is given by
	\begin{equation}
		\label{eq:sodium_update_rule} 
		\begin{split}
		\dot{\hat{\theta}}^\Na &= \gamma_\Na P_\Na \Psi_\Na(v-\hat{v}) \\
		\dot{\Psi}_\Na &= -\gamma_\Na \Psi_\Na + \Phi_\Na(v,\hat{w}^\Na) \\
		\dot{P}_\Na &= \alpha_\Na P_\Na
		- \alpha_\Na P_\Na \Psi_\Na \Psi_\Na^\transp P_\Na
		\end{split}	
	\end{equation}
	For $\Psi_\Na(0)=0$ and diagonal $P_\Na(0)\succ 0$, the matrices $\Psi_\Na(t)$, and $P_\Na(t)$ are diagonal for all $t \ge 0$, and \eqref{eq:sodium_update_rule} becomes the network-distributed update rule
	\begin{equation*}
		\begin{split}
		\dot{\hat{\mu}}_i^\Na 
		&= \gamma_\Na p^\Na_i \psi^\Na_i (v_i-\hat{v}_i) \\
		\dot{\psi}^\Na_i &= -\gamma_\Na \psi^\Na_i - (\hat{m}_i^\Na)^3\hat{h}_i^\Na(v_i-\nernst^\Na) \\
		\dot{p}_i^\Na &= \alpha_\Na p_i^\Na -
		\alpha_\Na (p_i^\Na)^2(\psi_i^\Na)^2
		\end{split}
	\end{equation*}
	for $i=1,2$. The same simplification applies to the update rules of $\hat{\theta}^\K$ and $\hat{\theta}^\GABA$. Hence the observer \eqref{eq:adaptive_observer_dist} becomes fully distributed with respect to neurons in the network and individual neuronal membrane currents in each neuron.
\end{ex}

\subsection{Numerical simulation}
\label{sec:numerical}

We finish this section by simulating online parameter estimation of the two-neuron HH network of \cref{ex:HH_network}. By doing so, we also show that the distributed adaptive observer estimates are able to approximately track slowly time-varying parameters. We consider a configuration in which the two neurons in the true system have identical parameters, except for their synaptic maximal conductances and control inputs. The parameters to be estimated are given by $\mu_1^\Na = \mu_2^\Na = 120$, $\mu_1^\K = \mu_2^\K = 36$,
\begin{align*}
	\mu_{1,2}^\GABA &= 0.75 - 0.4(1+e^{-(t-750)/100})^{-1},
	\quad \text{and}
	 \\
	\mu_{2,1}^\GABA &= 0.25 + 0.4(1+e^{-(t-750)/100})^{-1},
\end{align*}
while the remaining (non-adaptive) parameters are described in \cref{app:parameters} (which also describes initial conditions). Figure \ref{fig:voltages} shows the resulting voltages of the network.

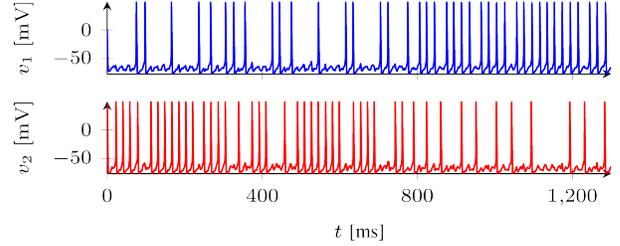
\begin{figure}[t]
	\centering
	\if\usetikz1
	\tikzset{png export}
	\begin{tikzpicture}
		\begin{groupplot}[
			group style={group size=1 by 2,
			vertical sep=10pt,
			},
			height=2.5cm,
			width=8.0cm, 
			filter discard warning=false,
			axis y line = left,
			axis x line = bottom,	
			tick label style={font=\scriptsize},	
			label style={font=\scriptsize},
			legend style={font=\footnotesize},
			xtick={0,400,800,1200},
			]				
			\nextgroupplot[
				ylabel={$v_1 \; \mathrm{[mV]}$},
				,xticklabels={,,}
				]
			\addplot[color=blue,semithick,smooth]
				table[x index=0,y index=1] 
				{./data/voltages_distr=false_samegains=true.txt};
			\nextgroupplot[
				ylabel={$v_2 \; \mathrm{[mV]}$},
				xlabel={$t$ [ms]}, 
				]				 
			\addplot [color=red,semithick,smooth]
				table[x index=0,y index=2] 
				{./data/voltages_distr=false_samegains=true.txt};
		\end{groupplot}
	\end{tikzpicture}
	\else
		\includegraphics[scale=0.125]{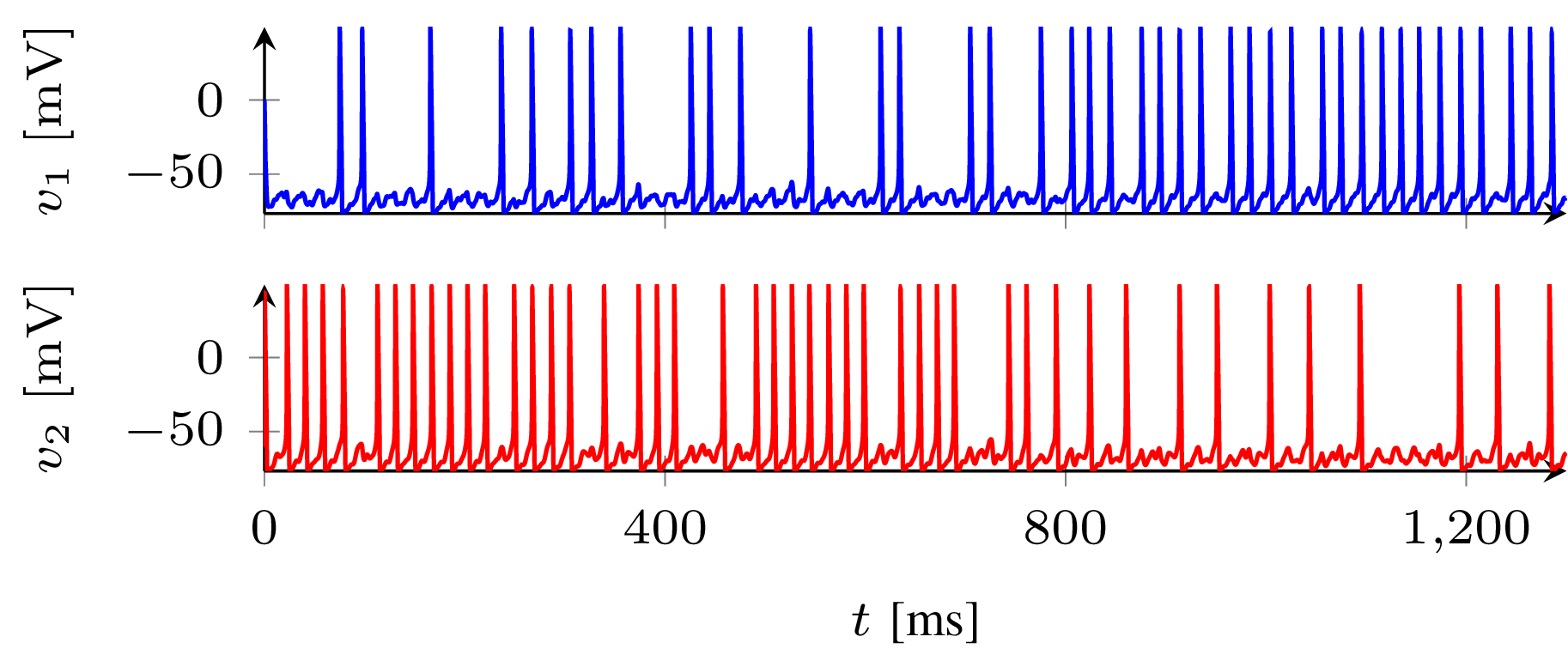}
	\fi
	\caption{Voltages of the two-neuron HH network in \cref{sec:numerical}.}
	\label{fig:voltages}
\end{figure}

We compare the performance of the distributed observer  \eqref{eq:adaptive_observer_dist}, implemented according to \cref{ex:network_distribution}, with the non-distributed version \eqref{eq:adaptive_observer}. To check the robustness of the observers to measurement noise, both observers are simulated with the measured $v(t)$ replaced by $v(t)+e(t)$, where $e(t)$ is white Gaussian measurement noise chosen so that the signal-to-noise between $v(t)$ and $e(t)$ is of $40$ dB.

The non-distributed observer is simulated with $\alpha = 0.15$ and $\gamma = 2$. The distributed observer is simulated with the two different gain sets below:
\begin{table}[h]
	\normalsize
	\centering
	\begin{tabular}{|p{.7cm}|p{.5cm}|p{.5cm}|p{.5cm}|p{.5cm}|p{.6cm}|p{.6cm}|p{.6cm}|}
	\hline
	Gains &$\gamma_0$ & $\gamma_\Na$ & $\gamma_\K$ & $\gamma_\GABA$ & 
	$\alpha_\Na$ & $\alpha_\K$ & $\alpha_\GABA$
	\\
	\hline
	(A) & $2$ & $2$ & $2$ & $2$ & $0.15$ & $0.15$ & $0.15$
	\\
	\hline			
	(B) & $2$ & $2$ & $2$ & $0.8$ & $0.15$ & $0.15$ & $0.03$
	\\
	\hline		
	\end{tabular}
\end{table}

The behaviour of the estimates $\hat{\mu}^\Na_i$ and $\hat{\mu}^\GABA_{i,j}$ is shown in Figures \ref{fig:gNa} and \ref{fig:gSyn} (the behaviour of $\hat{\mu}^\K_i$ is qualitatively similar to that of $\hat{\mu}^\Na_i$ and is omitted). The estimates of the non-distributed observer converge very rapidly to a neighborhood of the true parameter values (after only a few spikes, compare with Figure \ref{fig:voltages}). This indicates a rather aggressive choice of gains, which can also be inferred from the perturbation in $\hat{\mu}^\GABA_{1,2}$ seen just before $1200$ ms (this perturbation occurs due to a momentary decrease in the excitation provided by $v_2$, which increases sensitivity to noise). For the distributed observer with gain set (A), which mimics the aggressive gains of the non-distributed observer, it can be seen that $\hat{\mu}^\Na_i$ converges with a small smooth transient, but $\hat{\mu}^\GABA_{i,j}$ converges with a large transient with rapid oscillations. A less aggressive choice of gains could mitigate the undesired transients and decrease the sensitivity to noise, at the cost of a slower convergence rate (for a detailed analysis, see \cite{burghi_adaptive_2021}). The distributed observer allows being strategic with respect to the choice of gains: the gain set (B) remedies the undesirable transients in $\hat{\mu}^\GABA_{i,j}$ with a less aggressive choice of gains for $\alpha_\GABA$ and $\gamma_\GABA$, while keeping the other gains from the set (A). As a result, the good convergence properties of $\hat{\mu}_i^\Na$ (and of $\hat{\mu}_i^\K$) are preserved.

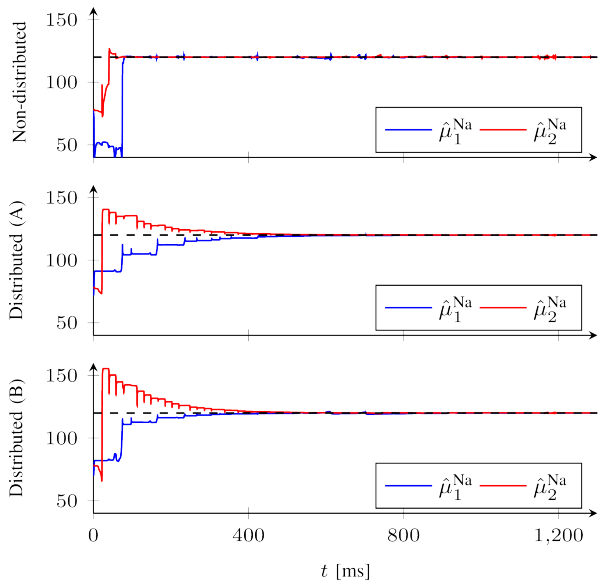
\begin{figure}[b]
	\centering
	\if\usetikz1
	\tikzset{png export}
	\begin{tikzpicture}
		\begin{groupplot}[
			group style={group size=1 by 3,
			vertical sep=10pt,
			},
			height=3.50cm,
			width=8.0cm, 
			axis y line = left,
			axis x line = bottom,	
			tick label style={font=\scriptsize},	
			label style={font=\scriptsize},
			legend style={font=\footnotesize},
			ymin=40,ymax=160,
			xtick={0,400,800,1200},
			legend pos=south east,
			legend columns = -1,	
			]
			\nextgroupplot[ylabel=Non-distributed,xticklabels={,,}]
			\addplot[color=blue,semithick,smooth]
				table[x index=0,y index=1] 
				{./data/parameters_distr=false_samegains=true.txt};
			\addlegendentry{$\hat{\mu}^\Na_1$};
			\addplot[color=red,semithick,smooth]
				table[x index=0,y index=2] 
				{./data/parameters_distr=false_samegains=true.txt};
			\addlegendentry{$\hat{\mu}^\Na_2$};
			\addplot[dashed,black,domain=0:1300,semithick]{120};
			
			\nextgroupplot[ylabel=Distributed (A),xticklabels={,,}]
			\addplot[color=blue,semithick,smooth]
				table[x index=0,y index=1] 
				{./data/parameters_distr=true_samegains=true.txt};
			\addlegendentry{$\hat{\mu}^\Na_1$};
			\addplot[color=red,semithick,smooth]
				table[x index=0,y index=2] 
				{./data/parameters_distr=true_samegains=true.txt};
			\addlegendentry{$\hat{\mu}^\Na_2$};
			\addplot[dashed,black,domain=0:1300,semithick]{120};
			
			\nextgroupplot[ylabel=Distributed (B),
							xlabel={$t$ [ms]},
							]
			\addplot[color=blue,semithick,smooth]
				table[x index=0,y index=1] 
				{./data/parameters_distr=true_samegains=false.txt};
			\addlegendentry{$\hat{\mu}^\Na_1$};
			\addplot[color=red,semithick,smooth]
				table[x index=0,y index=2] 
				{./data/parameters_distr=true_samegains=false.txt};
			\addlegendentry{$\hat{\mu}^\Na_2$};
			\addplot[dashed,black,domain=0:1300,semithick]{120};
			
		\end{groupplot}
	\end{tikzpicture}
	\else
		\includegraphics[scale=0.125]{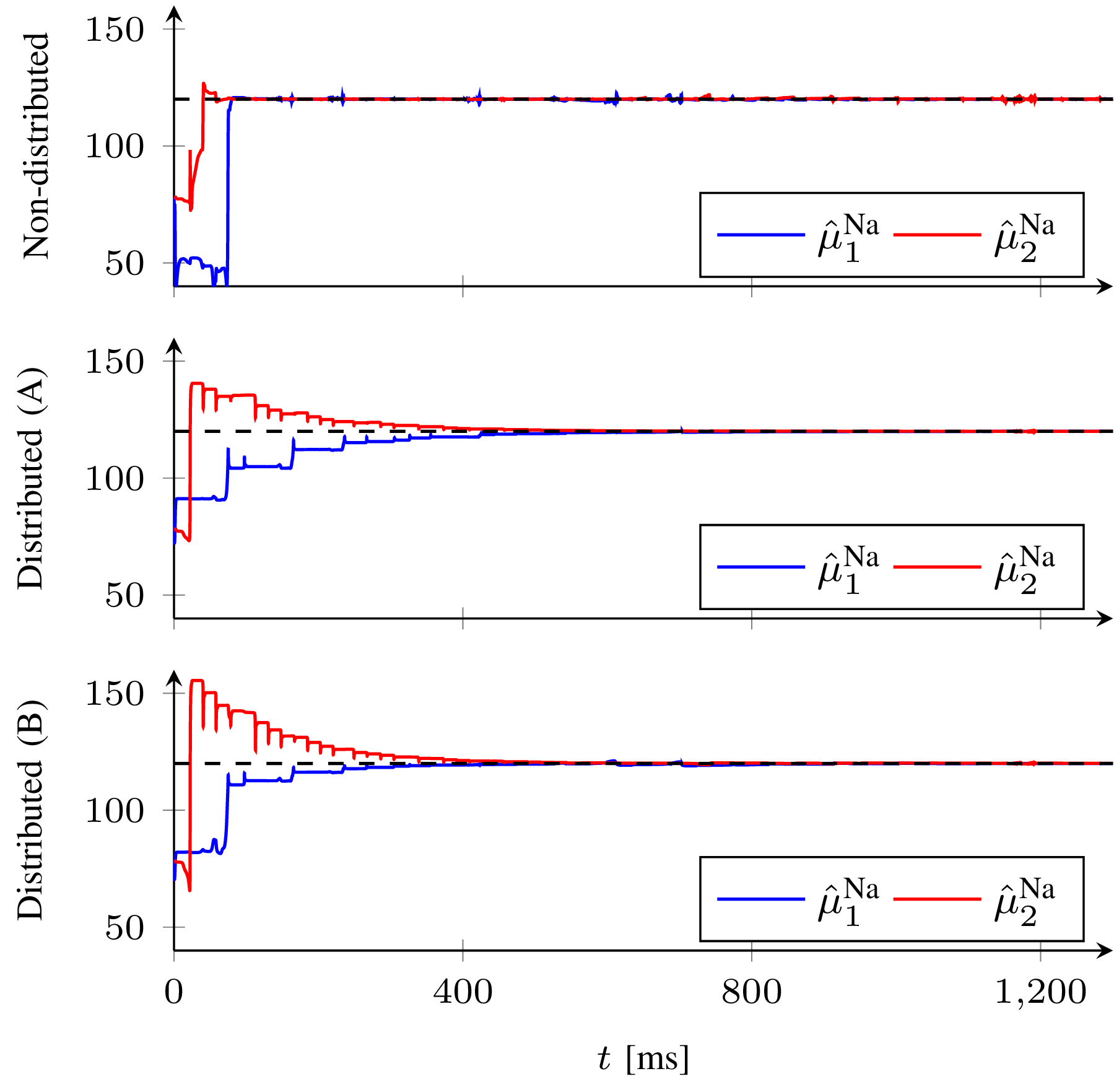}
	\fi
	\caption{Estimates of $\mu^\Na_1=\mu^\Na_2$ (dashed line) for the non-distributed and distributed adaptive observers with gain sets (A) and (B).}
	\label{fig:gNa}
\end{figure}

\begin{figure}[b]
	\centering
	\if\usetikz1
	\tikzset{png export}
	\begin{tikzpicture}
		\begin{groupplot}[
			group style={group size=1 by 3,
			vertical sep=10pt,
			},
			height=3.50cm,
			width=8.0cm, 
			axis y line = left,
			axis x line = bottom,	
			tick label style={font=\scriptsize},	
			label style={font=\scriptsize},
			legend style={font=\footnotesize},
			ymin=-0.1,ymax=1.1,
			xtick={0,400,800,1200},
			legend pos=north east,
			legend columns = -1,	
			]
			\nextgroupplot[ylabel=Non-distributed,xticklabels={,,}]
			\addplot[color=blue,semithick,smooth]
				table[x index=0,y index=5] 
				{./data/parameters_distr=false_samegains=true.txt};
			\addlegendentry{$\hat{\mu}^\GABA_{1,2}$};
			\addplot[color=red,semithick,smooth]
				table[x index=0,y index=6] 
				{./data/parameters_distr=false_samegains=true.txt};
			\addlegendentry{$\hat{\mu}^\GABA_{2,1}$};
			\addplot[dashed,black,domain=0:1300,semithick]{0.75 - 0.4/(1+exp(-(x-750)/100)};
			\addplot[dashed,black,domain=0:1300,semithick]{0.25 + 0.4/(1+exp(-(x-750)/100)};
			
			\nextgroupplot[ylabel=Distributed (A),xticklabels={,,}]
			\addplot[color=blue,semithick,smooth]
				table[x index=0,y index=5] 
				{./data/parameters_distr=true_samegains=true.txt};
			\addplot[color=red,semithick,smooth]
				table[x index=0,y index=6] 
				{./data/parameters_distr=true_samegains=true.txt};
			\addplot[dashed,black,domain=0:1300,semithick]{0.75 - 0.4/(1+exp(-(x-750)/100)};
			\addplot[dashed,black,domain=0:1300,semithick]{0.25 + 0.4/(1+exp(-(x-750)/100)};
			
			\nextgroupplot[ylabel=Distributed (B),
							xlabel={$t$ [ms]},
							]
			\addplot[color=blue,semithick,smooth]
				table[x index=0,y index=5] 
				{./data/parameters_distr=true_samegains=false.txt};
			\addplot[color=red,semithick,smooth]
				table[x index=0,y index=6] 
				{./data/parameters_distr=true_samegains=false.txt};
			\addplot[dashed,black,domain=0:1300,semithick]{0.75 - 0.4/(1+exp(-(x-750)/100)};
			\addplot[dashed,black,domain=0:1300,semithick]{0.25 + 0.4/(1+exp(-(x-750)/100)};
			
		\end{groupplot}
	\end{tikzpicture}
	\else
		\includegraphics[scale=0.125]{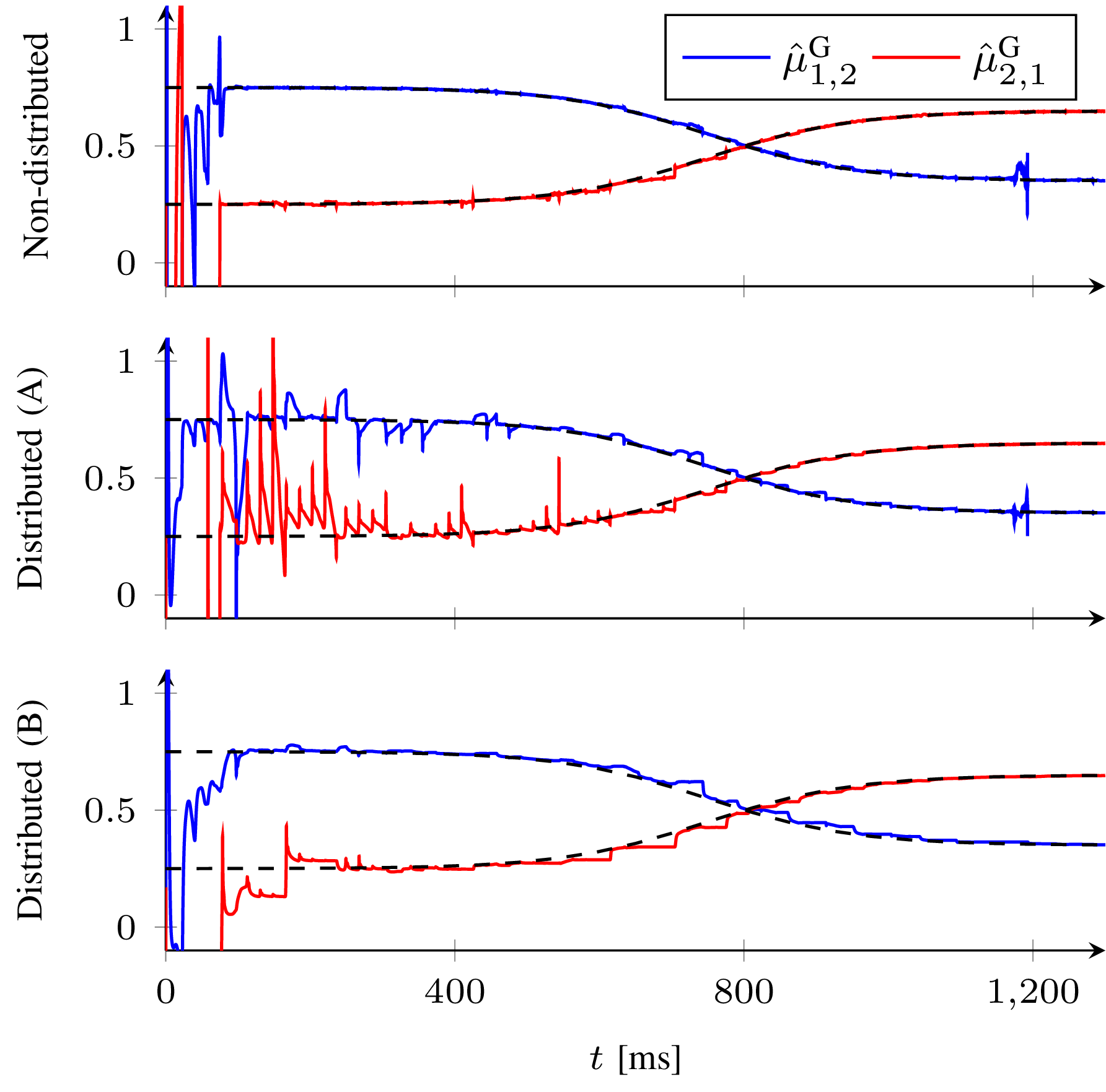}
	\fi
	\caption{Estimates of $\mu^\GABA_{1,2}$ and $\mu^\GABA_{2,1}$ (dashed lines) for the non-distributed and distributed adaptive observers with gain sets (A) and (B).}
	\label{fig:gSyn}
\end{figure}

\section{Discussion}
\label{sec:discussion}

We have shown that the distributed adaptive observer proposed in this paper is well-suited for estimating networks of biophysical neuronal models in real-time. We envision two applications of this work. First, \textit{experiment design in electrophysiology}: neuromodulators are capable of changing the maximal conductances of living neurons, and this process can be studied \textit{in vitro} using a closed-loop technique called dynamic clamp \cite{sharp_dynamic_1993}. The adaptive observer provides a means to track neuromodulatory changes in real-time, which allows incorporating this information in experiment design. The second application is the \textit{detection of qualitative changes in excitable regimes} \cite{tang_-line_2019}: epileptic seizures are associated to hyperexcitable neuronal states; given intracellular voltage data from a large network of neurons, the distributed adaptive observer provides a tractable means to keep track of the conductances responsible for modulating cellular excitability. This information is valuable for predicting the onset of switches in excitable regime associated to seizures.

In this paper, we have assumed the internal dynamics \eqref{eq:dw_true} is 
known. To address the case of an uncertain or unknown internal dynamics, a distributed version of the locally convergent nonlinearly parameterized adaptive observer in \cite{burghi_adaptive_2021}, designed to estimate the internal dynamics, can be proposed. This will be the subject of future work.

\appendix

\subsection{Proof of \cref{thm:convergence}}
\label{sec:proof_convergence}

The proof is inspired by the results of
\cite{aeyels_exponential_1999,jouffroy_relaxed_2003}  who provided sufficient conditions for exponential stability of nonlinear systems without requiring a Lyapunov function with a negative semidefinite derivative.
The proof also uses the idea of \textit{virtual system} from contraction theory, see \cite{lohmiller_contraction_1998}. The idea is to construct a dynamical system (the virtual system) whose trajectories contain the trajectories of the true system \eqref{eq:system} as well as those of the system \eqref{eq:dv_dist}-\eqref{eq:dtheta_dist}, and then show that the virtual system is exponentially contracting on a positively invariant set. Contraction of the virtual system's trajectories then imply exponential convergence of $(\vest,\west,\thetaest)$ to $(v,w,\theta)$.
The virtual system is given by
\begin{equation}
	\label{eq:virtual_system} 
\begin{split}
		\dot{s} &= \fw(\ytru,s) 		
		\\
		\dot{r} &= 
		\fvvirt(t,s,\eta) + \bv(\ytru,s,\utru) +
		(\gamma_0 I + \Psi^\transp \Pdiag \Gamma \Psi)(\ytru-r)  
		\\
		\dot{\eta} &= \Pdiag \Gamma \Psi (\ytru - r)
		\end{split}
\end{equation}
where
\[
	\fvvirt(t,s,\eta) = 
	\Phi^\transp(\ytru,\west,\utru)\eta +
	(\Phi^\transp(\ytru,s,\utru) - \Phi^\transp(\ytru,\west,\utru))\thetatru
\]
and where $\Psi$ comes from \eqref{eq:dPsi_dist}, which can be rewritten as
\begin{equation*}
	\dot{\Psi} = -\Gamma \Psi  + \Phi(v,\hat{w},u), \quad
	\Psi(0) = 0 \;.
\end{equation*}

Here, $s$, $r$ and $\eta$ are the virtual system states, as opposed to $\ytru(t)$,  $\west(t)$, and $\utru(t)$, which are treated as time-varying signals.
Notice that any solutions $\col(\wtru,\vtru,\thetatru)$ of \eqref{eq:system} and $\col(\west,\vest,\thetaest)$ of \eqref{eq:adaptive_observer_dist} are particular solutions of the virtual system. 
The first two equations of the virtual system have been written in a different order with respect to those of the true system and the adaptive observer, to simplify the notation in the remainder of the proof.


To proceed, we derive the \textit{differential} virtual system, given by
\begin{equation}
	\label{eq:diff_virtual_system} 
	\begin{pmatrix}
	\dot{\delta s} \\ \dot{\delta r} \\ 
	\dot{\delta \eta}
	\end{pmatrix}
	=
	\underbrace{
	\begin{bmatrix}
	J_{1,1} 
	& 0 \\ J_{2,1} & J_{2,2}
	\end{bmatrix}
	}_{\displaystyle J}
	\begin{pmatrix}
	\delta s \\ \delta r \\ \delta \eta
	\end{pmatrix}
\end{equation}
where the Jacobian $J$ has the components
\begin{subequations}
	\label{eq:Jvirt} 
	\begin{align}
	J_{1,1} &= \diag\{\partial_{s^1}\fw_1(v,s^1),\dotsc,\partial_{s^m}\fw_m(v,s^m)\}
	\\[.5em]
	J_{2,1} &= 
	\begin{bmatrix}
		\partial_s (\Phi^\transp(\ytru,s,\utru)\theta+
		\bv(\ytru,s,\utru)) \\ 0
	\end{bmatrix},
	\\[.5em]
	\label{eq:J22} 
	J_{2,2} &= 
	\begin{bmatrix}
		- (\gamma_0 I+\Psi^\transp \Pdiag \Gamma \Psi) &
		\Phi^\transp(\ytru,\west,\utru) \\
		-		\Pdiag \Gamma \Psi & 0
	\end{bmatrix}
	\end{align}
\end{subequations}
Following the ideas of  \cite{lohmiller_contraction_1998,jouffroy_relaxed_2003} we show that the dynamics  
\eqref{eq:virtual_system} are exponentially contracting on the invariant set $W \times \setreal^{\nv} \times \setreal^{\ntheta}$
using an infinitesimal coordinate transformation (we write $W:=\prod_j W_j$). We define
\begin{equation}
	\label{eq:coordinate_transformation} 
	\delta z = 
	\col
	\left(
		\mu \bar{\Theta}(t) \delta s,\quad
		\Theta_0(t) \begin{pmatrix}
		\delta r \\
		\delta \eta
		\end{pmatrix}
	\right),
\end{equation}
where $\mu > 0$ is an arbitrary constant, $\bar{\Theta}$ is given by
\[
	\bar{\Theta}(t) := \diag\{\Theta_1,\dotsc,\Theta_m\},
\]
with the $\Theta_j(t)$ from \cref{assum:int_dyn_contraction}, and $\Theta_0(t)$ is given by 
\begin{equation*}
	\Theta_0(t) = 
	\begin{bmatrix}
		I & -\Psi^\transp \\
		0 & \Rdiag^{\frac{1}{2}}
	\end{bmatrix},
\end{equation*}
with 
\[ \Rdiag(t) = \diag\{R_1(t),\dotsc,R_m(t)\}:= \Pdiag(t)^{-1} .\]
Under \cref{assum:true_invariant_set,assum:int_dyn_contraction},
$\Psi(t)$ is bounded for all $t\ge 0$, which implies $\Rdiag(t)$ is a well-defined upper-bounded inverse of $\Pdiag(t)$, see \cite[Lemma 3]{burghi_adaptive_2021}. Furthermore, part \textit{(i)} of \cref{assum:required_assumption} ensures that
\begin{equation}
	\label{eq:lower_bound_R} 
	 R_j(t) \succeq \delta_j \alpha_j  e^{-2\alpha_j T} I, \quad t \ge T
\end{equation}
and $R_j(t) \succeq e^{-\alpha_j T} R(0) \succ 0$ for $t\in[0,T)$, see 
\cite[Lemma 1]{zhang_adaptive_2001}.
As a consequence, $J$ in \eqref{eq:diff_virtual_system} is bounded on $W \times \setreal^{\nv} \times \setreal^{\ntheta}$ for all $t \ge 0$, and $\Theta_0(t)$ is uniformly invertible and bounded, with
\begin{equation}
	\label{eq:Thetainv} 
	\Theta_0^{-1}(t) = 
	\begin{bmatrix}
	I & \Psi^\transp \Pdiag^{\frac{1}{2}} \\
	0 & \Pdiag^{\frac{1}{2}}
	\end{bmatrix}
\end{equation}

Following \cite{jouffroy_relaxed_2003}, we prove contraction of the virtual system \eqref{eq:virtual_system} by showing that $\|\delta z(t)\| \to 0$ as $t \to 0$. For that purpose, we now derive the transformed differential system. Taking the derivative of  \eqref{eq:coordinate_transformation} and using the inverse coordinate transformation, we obtain
\begin{equation}
\label{eq:F} 
	\dot{\delta z}
	=
	\underbrace{
	\begin{bmatrix}
		\bar{F} & 0 \\
		\mu^{-1} J_{2,1} 
		\bar{\Theta}^{-1} & F_0
	\end{bmatrix}
	}_{\displaystyle F(t,s)}
	\delta z
\end{equation}
where $\bar{F} = \diag\{F_1,\dotsc,F_m\}$, with $F_j$ given by \eqref{eq:Mint}, and where $F_0$ is given by
\begin{equation}
	\label{eq:gen_jacobian} 
	F_0 = \big(\dot{\Theta}_0+\Theta_0 J_{2,2}\big)\Theta_0^{-1}
\end{equation}
(here, we have used $\Theta_0 J_{2,1} = J_{2,1}$).

The dynamics of the distance $\|\delta z\|$ is governed by
\[
	\frac{d}{dt}\|\delta z(t)\|^2 
	= 
	\delta z^\transp (F+F^\transp) \delta z
\]
and hence we seek an upper bound for $F+F^\transp$. To compute $F_0 + F_0^\transp$, it is worth noticing that since $d(P_j^{-1})/dt = -P_j^{-1}\dot{P}_jP_j^{-1}$, the diagonal elements of $\Rdiag$ obey
\begin{equation*}
	\dot{R}_j = -\alpha_j R_j + \alpha_j \Psi_j\Psi_j^\transp
\end{equation*}
from where we derive the relation
\begin{equation}
\label{eq:relation_dR} 
	\tfrac{d}{dt}(\Rdiag^{\frac{1}{2}})\Pdiag^{\frac{1}{2}} + \Pdiag^{\frac{1}{2}} \tfrac{d}{dt}(\Rdiag^{\frac{1}{2}}) = -A + \Pdiag^{\frac{1}{2}} AD \Pdiag^{\frac{1}{2}}
\end{equation}
It follows from \eqref{eq:J22}, \eqref{eq:Thetainv}, \eqref{eq:gen_jacobian} and \eqref{eq:relation_dR} that
\[
\resizebox{\hsize}{!}{$
	F_0+F_0^\transp = 
	\begin{bmatrix}
	-2\gamma_0 I & -\gamma_0 \Psi^\transp \Pdiag^{\frac{1}{2}} \\
	-\gamma_0 \Pdiag^{\frac{1}{2}} \Psi &
	-A + \Pdiag^{\frac{1}{2}}(AD-\Gamma \Psi \Psi^\transp - \Psi \Psi^\transp \Gamma)\Pdiag^{\frac{1}{2}}
	\end{bmatrix}
	$}
\]

Now, using \cref{assum:int_dyn_contraction} and  $F$ in \eqref{eq:F}, we have that
\[
 F+F^\transp \le \varepsilon I + M + N
\]
where
\[
	\begin{split}
	M &= 
	\begin{bmatrix}
		-\varepsilon I & \mu^{-1} \bar{\Theta}^{-\transp}J^\transp_{2,1} \\[.5em]
		\mu^{-1} J_{2,1} 
		\bar{\Theta}^{-1} &
		-\varepsilon I + \gamma_0 M_{2,2}
	\end{bmatrix} \\[.5em]
	M_{2,2} &= 		\begin{bmatrix}
		-I & -\Psi^\transp \Pdiag^{\frac{1}{2}}\\
		-\Pdiag^{\frac{1}{2}} \Psi &
		- 
		\Pdiag^{\frac{1}{2}}
		\Psi\Psi^\transp
		\Pdiag^{\frac{1}{2}}
		\end{bmatrix}
	\end{split}
\]
with $\varepsilon >0$ an arbitrary constant, and 
\begin{equation}
	\label{eq:N} 
	N = 
	\begin{bmatrix}
		-\min_j\{\lambda_j\} I & 0 & 0 \\
		0 & -\gamma_0 I & 0 \\
		0 & 0 & -A + \Pdiag^{\frac{1}{2}} Q \Pdiag^{\frac{1}{2}}
	\end{bmatrix}
\end{equation}
with
\[
Q = AD + \gamma_0 \Psi\Psi^\transp 
		-\Gamma \Psi \Psi^\transp - \Psi \Psi^\transp \Gamma
\]
In what follows, we choose $\varepsilon$ such that 
\begin{equation}
	\label{eq:epsilon} 
0 < \varepsilon < \min\{\gamma_0,\lambda_1,\dotsc,\lambda_m,\beta\}
\end{equation}
where $\beta > 0$ is taken from part (\textit{ii}) of \cref{assum:required_assumption}.

Using Schur's complement,
we see that $M_{2,2} \preceq 0$. Then, since $\mu>0$ is arbitrary and $\bar{\Theta}$ and $J_{2,1}$ are bounded on on $W \times \setreal^{\nv} \times \setreal^{\ntheta}$, uniformly in $t \ge 0$, we can use Schur's complement again to show that $M \preceq 0$ for a sufficiently large choice of $\mu(\varepsilon) > 0$. Hence for that choice it follows that  
\[
	\begin{split}
	\frac{d}{dt}\|\delta z(t)\|^2 
	&\le  
	\delta z(t)^\transp (N(t)+\varepsilon I) \delta z(t)\\
	&\le (\lambda_{\max}(N(t)) +\varepsilon)
	\|\delta z(t)\|^2
	\end{split}
\]
Solving for this inequality at every initial time $t\ge T$ yields
\[
	\|\delta z(t+T)\| \le 
	\|\delta z(t)\| \exp\big(\frac{1}{2}\int_t^{t+T}
	(\lambda_{\max}(N(\tau))+\varepsilon) d\tau\big)
\]

Now since $A$, $\Pdiag$ and $Q$ are all symmetric, \cite[Theorem 4.5.9]{horn_matrix_2012} and \cite[Corollary 4.3.15]{horn_matrix_2012} can be used to show that
\[\begin{split}
	\lambda_{\max}(-A+\Pdiag^{\frac{1}{2}}Q\Pdiag^{\frac{1}{2}}) 
	&\le \lambda_{\max}(-A) + \lambda_{\max}(\Pdiag)\lambda_{\max}(Q)\\
	&\le 
	-\underline{\alpha} + 
	(\min_j\big(\frac{\delta_j \alpha_j}{e^{2\,\alpha_j\,T}}\big))^{-1}
	\lambda_{\max}(Q)
	\end{split}
\]
for all $t \ge T$, where $\underline{\alpha} = \min\{\alpha_1,\dotsc,\alpha_m\}$, and where we have used $\lambda_{\max}(\Pdiag) = (\lambda_{\min}(\Rdiag))^{-1}$ as well as \eqref{eq:lower_bound_R}.

It finally follows from part (\textit{ii}) of  \cref{assum:required_assumption}, and the form of $N$ in \eqref{eq:N}, that for all $t \ge T$, there is a $T$ such that 
\[
\int_t^{t+T}
	(\lambda_{\max}(N(\tau))+\varepsilon) d\tau
	\le (-\min\{\gamma_0,\lambda_j,\beta\}+\varepsilon)T
\]
Given our choice of $\varepsilon$ in \eqref{eq:epsilon}, the right hand side above is strictly negative. Hence, similarly to \cite[Corollary 3.1]{jouffroy_relaxed_2003}, we conclude that $\|\delta z(t)\| \to 0$ as $t \to 0$ and that the virtual system \eqref{eq:virtual_system} is exponentially contracting on the invariant set $W \times \setreal^{\nv} \times \setreal^{\ntheta}$. As a result, we have $\col(\hat{w},\hat{v},\hat{\theta}) \to \col(w,v,\theta)$ as $t\to\infty$ and the result is proven.

\subsection{Simulation parameters}
\label{app:parameters} 
The parameters in $\Phi(v,w)$ and $\bv(v,u)$ are given by $\mu_1^\Leak = \mu_2^\Leak = 0.3$, $\nernst^\Na = 55$, $\nernst^\K = -77$, $\nernst^\Leak = -54.4$, and $\nernst^\GABA = -80$; the parameters of the intrinsic gating variable dynamics \eqref{eq:m_dyn}-\eqref{eq:h_dyn} are given in the table in \cite[Appendix C.1]{burghi_adaptive_2021}; the parameters of the synaptic gating variable dynamics \eqref{eq:s_dyn} are given by $a_\syn = 2$, $b_\syn = 0.1$, $\rho_\syn = -45$ and $\kappa_\syn = 2$. The two neurons are excited with the control inputs 
$
		u_1(t) = 2 + \sin(2\pi t /10) + \sin(2\pi t/7) + \sin(2\pi t/4)$ and 
		$u_2(t) = 1+2\sin(2\pi t/9) + \sin(2\pi t/5)$. All simulations are performed using the Euler-Maruyama method with $dt = 10^{-4}$ ms. The initial conditions are given by $v_1(0) = \hat{v}_1(0) = 0$, $v_2(0)=\hat{v}_2(0) = -60$, 
$w^\Na(0) = (0,0.5,0,0.5)^\transp$, $w^K(0) = (0,0.5)^\transp$, $w^\GABA = (0,0.5)^\transp$, $\hat{w}^\Na(0) = (0.5,0,0.5,0)^\transp$, $\hat{w}^K(0) = (0.5,0)^\transp$, $\hat{w}^\GABA = (0.5,0)^\transp$, $\hat{\theta}^\Na = \hat{\theta}^\K = (78,78)^\transp$, $\hat{\theta}^\GABA = (0,0)^\transp$.		
	

\bibliographystyle{IEEEtran}
\bibliography{references}

\end{document}